\pdfoutput=1
\documentclass[11pt]{article}
\usepackage{amsmath,amsthm,amssymb,color,verbatim,graphicx}
\newcommand{\remove}[1]{}
\newtheorem{thm}{Theorem}[section]
\newtheorem{claim}[thm]{Claim}
\newtheorem{lem}[thm]{Lemma}
\newtheorem{define}[thm]{Definition}
\newtheorem{cor}[thm]{Corollary}

\newtheorem{THM}{Theorem}

\newtheorem{fact}[thm]{Fact}

\renewcommand{\remove}[1]{}

\newcommand{\eps}{{\varepsilon}}
\renewcommand{\l}{\left}
\renewcommand{\r}{\right}

\newcommand{\de}{{\delta}}

\newcommand{\comments}[1]{}

\newcommand{\SO}{\mathbf{SO}}
\newcommand{\SE}{\mathbf{SE}}
\renewcommand{\L}{\mathcal{L}}
\renewcommand{\S}{\mathbb{S}}
\newcommand{\Di}{\mathbf{Disc}}
\newcommand{\D}{\mathbf{D}}
\renewcommand{\P}{\mathcal{P}}

\newcommand{\R}{\mathbb{R}}

\newcommand{\T}{\mathbf{T}}
\newcommand{\C}{\mathcal{C}}
\newcommand{\U}{\mathcal{U}}
\renewcommand{\O}{\mathcal{O}}
\newcommand{\X}{\mathcal{X}}
\newcommand{\Y}{\mathcal{Y}}

\newcommand{\zo}{\{0, 1\}}

\renewcommand{\S}{\mathbb{S}}

\def\draft{0}   

\ifnum\draft=1 
    \def\ShowAuthNotes{1}
\else
    \def\ShowAuthNotes{0}
\fi

\ifnum\ShowAuthNotes=1
\newcommand{\authnote}[2]{{ \footnotesize \bf{\color{red}[#1's Note: {\color{blue}#2}]}}}
\else
\newcommand{\authnote}[2]{}
\fi

\makeatletter
\def\sand{%
  \end{tabular}%
  \hskip 0.5em \@plus.17fil\relax
  \begin{tabular}[t]{c}}
\makeatother
\begin{document} 


\title{\textbf{On Low Discrepancy Samplings in Product Spaces of Motion Groups}}
\author{Chandrajit Bajaj\thanks{Department of Computer Science and 
The Institute of Computational Engineering and Sciences,
Center for Computational Visualization, The University of Texas at Austin. Email: bajaj@cs.utexas.edu. Research supported in part by NSF Grant OCI-1216701 and SNL contract no. 1439100} \sand Abhishek Bhowmick\thanks{Department of Computer Science, The University of Texas at Austin. Email: bhowmick@cs.utexas.edu. Research supported in part by NSF Grant CCF-1218723.} \sand Eshan Chattopadhyay\thanks{Department of Computer Science, The University of Texas at Austin. Email: eshanc@utexas.edu.  Research supported in part by NSF Grant CCF-1218723.} \sand David Zuckerman\thanks{Department of Computer Science, The University of Texas at Austin. Email: diz@cs.utexas.edu. Research supported in part by NSF Grant CCF-1218723.}}
\maketitle

\thispagestyle{empty}

\begin{abstract}
Deterministically generating near-uniform point samplings of  the motion groups like $\SO(3)$, $\SE(3)$ and their $n$-wise products $\SO(3)^n$, $\SE(3)^n$ is fundamental to numerous applications in computational and data sciences. 
The natural measure of sampling quality is discrepancy.
In this work, our main goal is construct low discrepancy deterministic samplings in product spaces of the motion groups. To this end, we develop a novel strategy (using a two-step discrepancy construction) that leads to an almost exponential improvement in size (from the trivial direct product). To the best of our knowledge, this is the first nontrivial construction for $\SO(3)^n, \SE(3)^n$ and the hypertorus $\mathbb{T}^n$.

We also construct new low discrepancy samplings of  $\S^2$ and $\SO(3)$. 
The central component in our construction for $\SO(3)$ is an explicit construction of 
$N$ points in $\S^2$ with discrepancy $\tilde{\O}(1/\sqrt{N})$ with respect to convex sets,
matching the bound achieved for the special case of spherical caps in \cite{ABD_12}.
We also generalize the discrepancy of Cartesian product sets \cite{Chazelle04thediscrepancy} to the discrepancy of local Cartesian product sets. 

The tools we develop should be useful
in generating low discrepancy samplings of  other complicated geometric spaces.
\end{abstract}\newpage
\setcounter{page}{1}

\section{Introduction}

Generating nearly uniformly random samples from rigid body motion groups, like $\SO(3), \SE(3),$etc. is fundamental to several applications in computational sciences\cite{baj3}.  One 
is in predicting protein-protein docking  where the search and scoring is over rigid and/or flexible motion spaces \cite{baj1}. Other applications include multi-dimensional correlations, molecular dynamics, quantum computation, Monte-Carlo search, functional optimization,  numerical integration \cite{chazelle_discrep_book,Chetelat_Chirikjian,Mitchell_2008,Niederreiter_1990,Radovic_Sobol_Tichy_1996,Wang_Sloan_2008,Yershova_Jain_Lavalle_Mitchell_2010}.

We seek a deterministic sampling, that is, a deterministic construction of a set of points that can serve as a good sample.
The key measure of the quality of such a sampling in most applications is discrepancy.

\begin{define}[Discrepancy] Let $\X$ be a collection of sets in the universe $\U$ and $P \subseteq U$ be a collection of points.
The \emph{discrepancy} of $P$ with respect to $\X$ is
\[ \D(P,\X)=\max_{X \in \X} \left|\frac{|P \cap X|}{|P|} -  \frac{\mu(X)}{\mu(\U)}\right|,\]
where $\mu$ denotes Lebesgue measure.
\end{define}
 
One reason discrepancy is fundamental is that the error in using a point sampling for numerical integration, as in quasi-Monte Carlo methods,
is closely related to the discrepancy.  For excellent books devoted to this subject,
see Niederreiter \cite{Niederreiter_1990}, Matousek \cite{geombook}, and Chazelle \cite{chazelle_discrep_book} .

The low discrepancy sampling techniques are simpler if the underlying domain is a simple  Cartesian product, 
as is the case for translational motion, but it becomes considerably more challenging for, say, the rotation space.
The space $\SO(3)$, or, the special orthogonal rotation group, stems out of rotations around the origin in three dimensional space. The group behaves like the real projective space $\mathbb{R}\mathbb{P}^3$. It can be represented as a $3$-$sphere$, $\S^3$ embedded in $\mathbb{R}^4$ with antipodal points identified. The elements of $\SO(3)$ are defined as $3 \times 3$ orthogonal matrices with determinant $1$. The group operation is multiplication of matrices. However, they are not as numerically stable as quaternions and are less often used. Unit quaternions can be thought of as elements in $\mathbb{R}^4$ with norm $1$. More precisely, $x=(x_1,x_2,x_3,x_4) \in \mathbb{R}^4$ or $x=x_1+x_2i+x_3j+x_4k$ where $||x||=1$. Quaternion representations are  very convenient to combine  two or more rotations. 
Yet another representation of motion spaces is the use of Euler angles. Here, one specifies rotations in $\mathbb{R}^3$ about the three angles, one for each axis denoting the angle by which it needs to be rotated. A more detailed discussion can be found in \cite{Mitchell_2008}. \emph{Hopf fibration} (or Hopf coordinates), introduced by Heinz Hopf in 1931 \cite{hopf}, is another way to represent $\SO(3)$ in terms of a local Cartesian product of $\S^1$ and $\S^2$ which are much easier to visualize. We describe \emph{Hopf fibration} in more detail later.

In a series of works \cite{Lind-Valle2003, Yersh-Valle2004,Mitchell_2008,Yershova_Jain_Lavalle_Mitchell_2010},  low discrepancy samplings were constructed in $\SO(3)$ with respect to the class of local Cartesian products of axis aligned grids in $\S^1$ and $\S^2$.  

Next, we have $\SO(2)$ which is the group of rotations in $2$ dimensions. This is a subgroup of $\SO(3)$. Combining both rotation and translation we have the group $\SE(3)$ which captures the most general space of rigid body motion in $3$ dimensions. It involves a rotation component $\SO(3)$ and a translational component $\R^3$. In addition to $\SO(3), \SO(2)$ there are other important subgroups of $\SE(3)$ which are widely used in physics, computational biology and computer science, namely $\S^1$ (the torus), $\S^2$ (the $2$ dimensional sphere), $\T(k)$, the translational group in $k (\leq 3)$ dimensions and $\SE(2)$, the group of rotations and translations in $2$ dimensions. 

Finally, a lot of applications require low discrepancy sampling in product spaces of these basic motion groups. More precisely, we are interested in sampling from groups of the form $\SO(3)^n$, $\SE(3)^n$ and the hypertorus ($\mathbb{T}^n (\mathbb{T}\equiv \S^1)$), to name a few.

\subsection{Results}
Our main result is on obtaining small sized low discrepancy sets of product spaces of motion groups. More precisely, we will be interested in the special Euclidean group in $3$ dimensions, $\SE(3)$. This is the group of translations and rotations in $3$ dimensions. We will next consider the various important subgroups of $\SE(3)$ that are of interest in rigid body kinematics and more generally in computational biology, physics and computer science. The subgroups are \begin{itemize}
\item $\T(k), k=1,2,3:$ The group of translations in $k$ dimensions
\item $\SO(2)\  (\text{or }\S^1):$ The group of rotations in $2$ dimensions
\item $\S^2:$ The group of rotations of the $z$-axis around the origin
\item $\SO(3)\ (\text{or }\S^3):$ The group of rotations in $3$ dimensions
\item $\SE(2):$ The group of translation and rotation in $2$ dimensions
\end{itemize}  

In the first part of the paper, we consider these basic rigid body motion groups and in the later part we construct low discrepancy sets for $n$-wise product spaces of these groups using a careful derandomization of the trivial exponential (in $n$) sized construction.


Our first two results involve construction of low discrepancy sets in $\SO(3)$ and $\S^2$ which we make precise below.

\begin{THM}\label{thm:so3}There exists an efficiently generated collection of points $P$ such that $\D(P,\C)=\O\l(\frac{\log^{2/3} N}{N^{1/3}}\r)$ where $\C$ is the class of local Cartesian convex sets.
\end{THM}

To obtain Theorem \ref{thm:so3}, we give an explicit construction of low discrepancy samplings in $\S^2$  with respect to the class of all convex sets. This in particular generalizes the results of \cite{bd12,ABD_12} where the target classes were latitude spherical rectangles (defined later) and spherical caps respectively.
\begin{THM}\label{thm:s2}There exists an efficient deterministic sampling of $N$ points in $\S^2$ with discrepancy $O(\sqrt{\log N/N})$ against all spherical convex sets.\end{THM}

Another component we develop to prove Theorem~\ref{thm:so3} is a generalization of a lemma about direct products of low discrepancy sets. See Theorem~\ref{thm:localcart} for details.

In the second part of the paper, we show a novel technique to construct deterministic samplings in $\U^n$ with respect to $\C^n$
with subexponential size in $n$ using low discrepancy sets in $\U$ with respect to $\C$.  
We state our final main theorem informally here.
\begin{THM}[Informal]If there is a low discrepancy ($\eps$) point set in $\U$ with respect to $\C$ of size $m$, then there is a low discrepancy ($\eps n$) point set in $\U^n$ with respect to $\C^n$ of quasipolynomial (in $n,m$) size.
\end{THM}

We get an almost exponential improvement in the size of the point collection (than what a trivial direct product would give) using the following idea. We use low discrepancy sets (of size $m$, say) in the basis groups and take a direct product of the sets to obtain a new set of size $m^n$. Now, we apply another level of discrepancy minimization to further choose a subset of the $m^n$ points to still ensure the right discrepancy. The right tool to make the second step work is a pseudorandom generator which we describe in detail in Section~\ref{sec:product}. Using the results from the first part of the paper, this leads to a host of nontrivial constructions for low discrepancy sets in product spaces of rigid motion groups.

\subsection{Organization}
After some preliminaries, we review existing discrepancy bounds for various subgroups of $SE(3)$ in Section~\ref{sec:basic}. We also prove two of the three main theorems here. In Section~\ref{sec:product} we state our third main theorem on lifting point sets to products of groups.

\subsection{Preliminaries and Notation}

We let $\mu$ denote the Lebesgue measure. 

By convex sets we mean closed convex sets unless otherwise mentioned. We reserve calligraphic letters to denote collection of sets.

We let $\stackrel{\sim}{=}$ denote a  homeomorphism between two spaces.

\remove{
Some combinatorial notions and metrics required for determining the quality of sampling. The universe  will be the  unit cube $\U$ equipped with the volume measure $\mu$ and the Euclidean metric $d()$. Let $\S$ denote the set of all spheres contained in $\U$. We naturally extend the definition of measures in product spaces as the product measure. 

\begin{define}[Dispersion] Let $A$ be a collection of points in the unit cube $\U$. We define the dispersion of $A$ to be the measure of the largest sphere inside $\U$ which doesn't intersect with $A$.

\end{define}

\begin{define}[Min distance of sample] Let $A=\{a_1,..,a_l\}$ be a collection of points in the unit cube $\U$. The minimum distance of $A$ is defined to be $$\min_d(A)= \min_{i \neq j } d(a_i,a_j)   $$
\end{define}

 \begin{define}[$\epsilon$-net] Let the universe be the unit cube $\U$ and $A$ be a collection of points. We say that $A$ is an $\epsilon$-net for the set of spheres $\S$ in $\U$ if for any sphere $S \in \S$ of volume at least $\epsilon$, $A \cap S \neq \phi$. 
\end{define}
}

\section{The basic groups}\label{sec:basic}

We start with the simpler groups and review known constructions of low discrepancy sets in them and use these as building blocks in the other groups.
\subsection{$\T(k):$ The group of translations in $k$ dimensions}
The set of all translations in $k$ dimensions with addition being the operation forms a group. For $k=1,2,3,$ it is a subgroup of $\SE(3)$. We will asume $k \leq 3$ here. Note that $\T(k)\stackrel{\sim}{=} \R^k$. The volume element is simply taken to be the Lebesgue measure in $\R^k$. This group is very well understood in terms of low discrepancy sampling. We will consider the normalized cube $[0,1]^k$. The best construction is captured by the following theorem.
\begin{thm}[Theorem 3.8, \cite{book:neid}]There is an efficient set $\P$ of $N$ points in $[0,1]^k$ that achieves $$\D(\P,\mathcal{C}_1)=\O\l(N^{-1}(\log N)^{k-1}\r),$$ $$\D(\P,\mathcal{C}_2)=\O\l(N^{-1/k}(\log N)^{1-1/k}\r),$$ where $\mathcal{C}_1$ and $\mathcal{C}_2$ are the set of axis aligned hyper-rectangles and convex sets respectively. 
\end{thm}

\subsection{$\SO(2):$ The group of rotations in $2$ dimensions}
The group of rotations in $2$ dimensions, $\SO(2)$ is homeomorphic to $\S^1$, the unit circle with volume element being the arc length. In such a case, the equi-distributed point set on the unit circle forms the best low discrepancy set.
\begin{lem}[\cite{book:neid}]\label{lem:s1}There is an efficient set $\P$ of $N$ points in $\S^1$ that achieves discrepancy $\O\l(N^{-1}\r)$ against the class of contiguous intervals. 
\end{lem}
We record the following simple generalization.
Define $$\S^1_{[\theta_1,\theta_2]} = \{ (\cos \theta,\sin \theta) : \theta \in [\theta_1,\theta_2]\}. $$
\begin{cor}[bounded range]\label{cor:s1} For all constants $\theta_1,\theta_2$ such that $0 \le \theta_1 < \theta_2 \le 2\pi$, there exists an efficient deterministic sampling of $N$ points in $\S^1_{[\theta_1,\theta_2]}$that achieves discrepancy $\O\l(N^{-1}\r)$ against all contiguous intervals. 
\end{cor} 

\subsection{$\S^2:$ The group of rotations of the $z$-axis around the origin}
In this subsection, we present a low discrepancy point collection for convex sets that matches the bound of \cite{ABD_12} who obtain a similar bound for the special case of spherical caps. This is the first contribution of this work.

In the following, we consider the unit sphere $\S^2$ in $\R^3$. A spherical polygon is formed by edges defined by great circles. We will also use a latitude spherical rectangle which is not exactly a spherical rectangle. It is a rectangle formed by two latitudes and two longitudes. A spherical convex set in $\S^2$ is a closed set contained strictly in a hemisphere such that given any two points in the set, the geodesic connecting the two points also lie in the set. 
\paragraph{Point samplings on $\S^2$.}
Let $P_1$ be the Sobol point collection \cite{Sobol67} (or Hammersley point set or any other discrepancy optimal point set against axis aligned rectangles) of size $N$ in the unit square. For details about such point sets, refer to \cite{book:neid}. Define the Lambert equal-area projection $\phi:\R^2 \to \R^3$ which was introduced in \cite{bd12} as \[\phi(x,y) = (2\sqrt{y - y^2}\cos(2\pi x),2\sqrt{y - y^2}\sin(2\pi x) ,1-2y)\] 
The area-preserving Lambert map can be visualized in the following way (see Figure~1). Consider the unit square $[0,1] \times [0,1]$ and expand it to $[-\pi,\pi] \times [-1,1]$ and roll it into a cylinder with unit radius and height $2$. Now, draw a sphere maximally inscribed in the cylinder with the centers of the sphere and the cylinder coinciding. Now, consider any point $(x,y)$ on the unit square. It gets ``stretched'' to a point on the cylinder surface. Now, draw a line radially inward perpendicular to the axis of the cylinder. Let the line hit the sphere surface at point $z$. This is the image of the Lambert map on $(x,y)$. Thus, rectangles on the unit square get mapped to latitude spherical rectangles on the surface of the sphere as in Figure~1. Then our proposed point collection $P$ is the set \begin{equation}\label{eq:areap}\{\phi(x,y) : (x,y) \in P_1\}.\end{equation} From \cite{ABD_12, bd12}, it follows that \begin{equation}\label{eq:areapbound}\D(P,\C) = \O\l(\frac{\log N}{N}\r),\end{equation} where $\C$ is the class of latitude spherical rectangles.

\begin{figure}[h]
\centering
\includegraphics[scale=0.5]{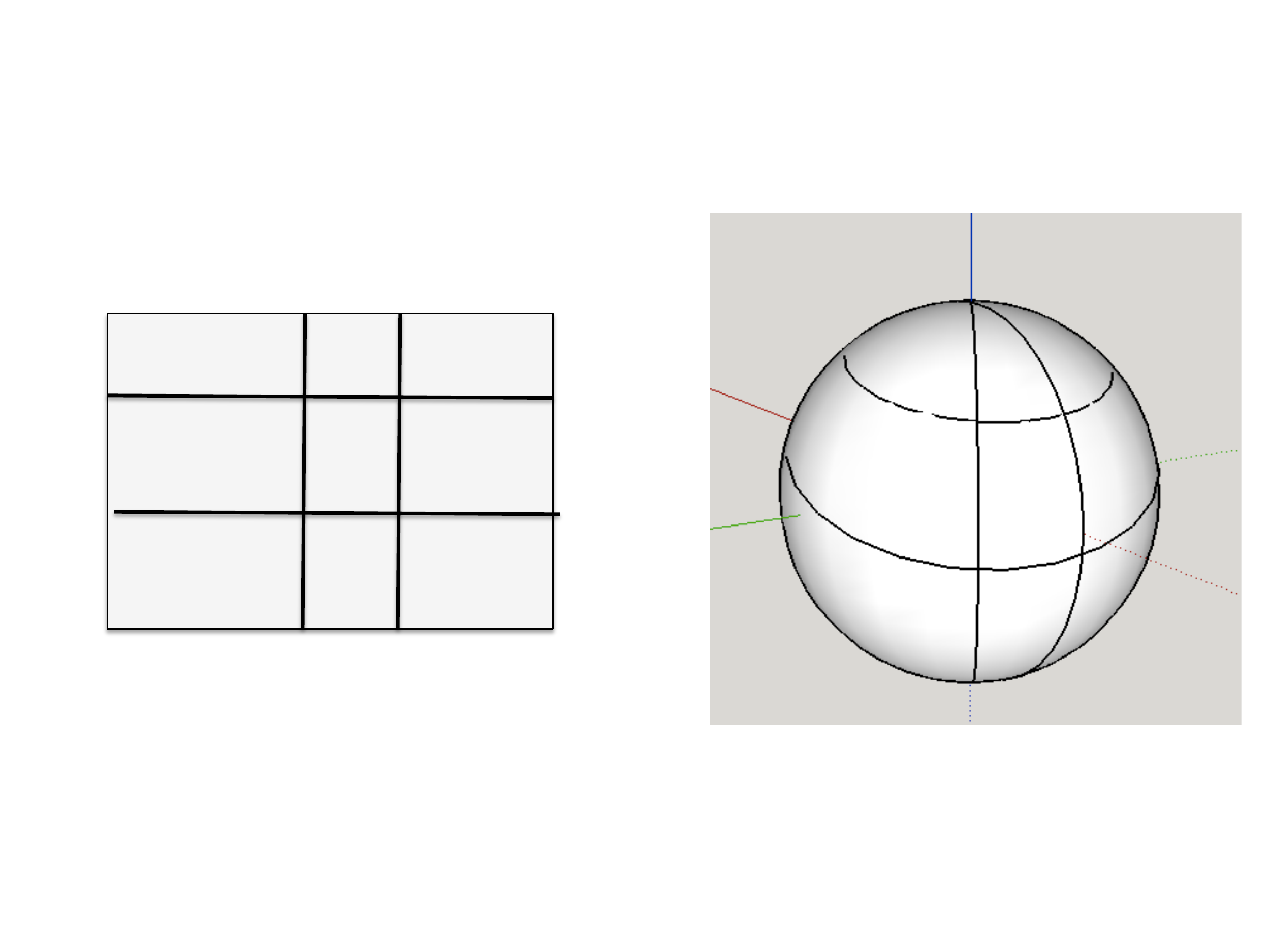}
\caption{The Lambert map: $\phi(x,y) = (2\sqrt{y - y^2}\cos(2\pi x),2\sqrt{y - y^2}\sin(2\pi x) ,1-2y)$}

\end{figure}

\paragraph{Discrepancy bounds for convex sets on $\S^2$.}
We state the main theorem of this section.
\begin{thm}\label{thm:mains2gen} Let $P$ be the point collection in $\S^2$ defined by \text{(\ref{eq:areap})}. Then $\D(P,\C_1)=\O\l(\sqrt{\D(P,\C_2)}\r)$ where  $\C_1$ and $\C_2$ are the class of spherical convex sets and latitude spherical rectangles.\end{thm}

As a corollary we first obtain one of the main theorems of this work.\\ \\
\textbf{Theorem \ref{thm:s2} (Restated). } There exists an efficient deterministic sampling of $N$ points in $\S^2$ with discrepancy $\O\l(\sqrt{\log N/N}\r)$ with respect to all convex sets.\\
\begin{proof}The proof follows from Theorem \ref{thm:mains2gen} and the Equation~\ref{eq:areapbound}.
\end{proof}

In the rest of this subsection, we write $\D(P)$ for the discrepancy of $P$ with respect to latitude spherical rectangles. As already noted above, we have $\D(P)=\O(\log N/N)$. We shall now use the proof technique in \cite{KuipNeid}  and first reduce the class of all spherical convex sets to that of spherical convex polygons. Here we appeal to spherical trigonometry. In the second part, we show that it is enough to guarantee discrepancy against latitude spherical rectangles which we already have in the previous paragraph.

We will need the following  spherical analogue of the separating hyperplane lemma, which we prove  via first principles. To the best of our knowledge, this has not been done before on $\S^2$. 
\begin{lem}\label{lem:hp} Let $C$ be a spherical convex set and $Q$  be any point in $\S^2$ outside $C$. Then there is a great circle  $\ell$ which separates $C$ and $Q$.  
\end{lem}

\begin{proof} Let $\ell_1$ be the geodesic connecting $Q$ to the nearest point, say $Q_1$, in $C$. This exists by convexity and closure. Draw great circle $\ell$ perpendicular to $\ell_1$ such that $Q$ and $Q_1$ are on opposite sides of $\ell$. We claim that $\ell$ separates $C$ and $Q$. 

For the sake of contradiction suppose there is some point $Q^{\prime} \in C$ such that $Q^{\prime}$ and $Q$ lie on the same side of $l$. Connect $Q^{\prime}$ and $Q_1$ by a geodesic $\ell_2$ and let $\ell_3$ be a geodesic passing through $Q$ and perpendicular to $\ell_2$ at $P_2$. Consider the spherical triangle $QQ_1Q_2$. By construction, we have $\angle QQ_2Q_1 = \pi/2$. Further, since $QQ_1$ is perpendicular to $\ell_1$, we have $  \theta = \angle QQ_1Q_2 < \pi/2$. Using the Sine rules for spherical triangles , we have $$ \frac{\sin \pi/2}{\sin |QQ_1|} = \frac{\sin \theta}{\sin |QQ_2|}$$ which implies $\sin |QQ_1| > \sin |QQ_2|$. Since $|QQ_i| \le \frac{\pi}{2}$ for $i=1,2$ and observing that the $sine$ function is an increasing function in the interval $(0,\frac{\pi}{2}]$, we have $|QQ_1| > |QQ_2|$.

This contradicts our assumption that $Q_1$ is the closest point to $Q$ in $C$.
\end{proof}

We now state the following lemma which reduces proving low discrepancy bounds for spherical convex sets to low discrepancy sets for spherical convex polygons.
\begin{lem}$\L(P)=\max_{X \in \mathcal{P}}\l|\frac{|P \cap X|}{|P|}-\frac{\mu(X)}{\mu(\S^2)}\r|$ where $\mathcal{P}$ denotes the class of all spherical convex polygons.
\end{lem}
\begin{proof}Given any spherical convex set $C$, we will show the existence of two spherical convex polygons, $Q_1$ and $Q_2$ such that $\mu(Q_1)\leq \mu(C) \leq \mu(Q_2)$ and $|Z \cap P|$ is the same for $Z \in \{C,Q_1,Q_2\}$. This is enough because then we have \[\l|\frac{|P \cap C|}{|P|}-\frac{\mu(C)}{\mu(\S^2)}\r|\leq \max \l(\l|\frac{|P \cap Q_1|}{|P|}-\frac{\mu(Q_1)}{\mu(\S^2)}\r|,\l|\frac{|P \cap Q_2|}{|P|}-\frac{\mu(Q_2)}{\mu(\S^2)}\r|\r)\] and this finishes the proof.
We now begin showing the existence of $Q_1$ and $Q_2$ as required above.

We begin with $Q_1$. Let $P'=P \cap C$. Let $Q_1$ be the spherical convex hull of $P'$. Clearly, $\mu(Q_1) \leq \mu(C)$. Also, note that $Q_1 \cap P=C \cap P$ as $Q_1$ is the spherical convex hull of $P'$.

We now show the existence of $Q_2$. Let $P''=P - C$. Applying Lemma \ref{lem:hp} to each point with respect to $C$, we get a great circle $\ell$. Let $Q_2$ be the intersection of all the hemispheres containing $C$ that are defined by the  above $\ell$'s. Clearly, we have $\mu(C) \leq \mu(Q_2)$. Also, by construction, we have $|P \cap C|=|P \cap Q_2|$.
This finishes the proof.
\end{proof}

We now turn to proving the main theorem of this section.

\begin{proof}[Proof of Theorem \ref{thm:mains2gen}] The idea of the proof is similar to the proof of the plane case but requires a more careful analysis. Given a spherical convex polygon $X$, we will sandwich it between two sets $Q_1$ and $Q_2$ where both $Q_1$ and $Q_2$ are unions of of latitude spherical rectangles. Thus, we will have $Q_1 \subseteq X \subseteq Q_2$. With this in place, observe that \[\l(\frac{|Q_1 \cap P|}{|P|}-\frac{\mu(Q_1)}{\mu(\S^2)}\r)+\l(\frac{\mu(Q_1)}{\mu(\S^2)}-\frac{\mu(X)}{\mu(\S^2)}\r) \leq \l(\frac{|P \cap X|}{|P|}-\frac{\mu(X)}{\mu(\S^2)}\r) \]\[\leq \l(\frac{|Q_2 \cap P|}{|P|}-\frac{\mu(Q_2)}{\mu(\S^2)}\r)+\l(\frac{\mu(Q_2)}{\mu(\S^2)}-\frac{\mu(X)}{\mu(\S^2)}\r)\]
Hence, we have \[\l|\frac{|P \cap X|}{|P|}-\frac{\mu(X)}{\mu(\S^2)}\r| \leq \max_{i=1,2}\l|\frac{|P \cap Q_i|}{|P|}-\frac{\mu(Q_i)}{\mu(\S^2)}\r|+\frac{\max_{i=1,2}|\mu(Q_i)-\mu(X)|}{\mu(\S^2)}\]

We now construct sets $Q_1$ and $Q_2$. Let $k>0$ be some integer parameter that we will choose later. Consider the hemisphere in which the spherical polygon $X$ lies. We draw $k$ equispaced latitudes and $k$ equispaced longitudes in this hemisphere. Thus the hemisphere is divided $k^{2}$ latitude spherical rectangles. We define $Q_1$ to be the union of all those latitude spherical rectangles which lie completely inside $X$ and $Q_2$ to be the union of all those latitude spherical rectangles which have non trivial intersection with $X$. We observe that by construction $Q_1 \subseteq X \subseteq Q_2$.

To bound the $\l|\frac{|P \cap Q_i|}{|P|}-\frac{\mu(Q_i)}{\mu(\S^2)}\r|$ for $i=1,2$, we need the following claim.
\begin{claim}
$Q_1$ and $Q_2$ defined above can be written as a disjoint union of at most $k$ latitude spherical rectangles. 
\end{claim}
\begin{proof}To prove this, let us fix two consecutive latitudes and consider the $k$ latitude spherical rectangles $\{ R_1,R_2,\cdots,R_k\}$ which are formed by these two latitudes and the $k$ longitudes. Notice that due to convexity of $Q_1$ there can be no gaps and hence there exists $1\le i\le j \le m$ such that $R_l \in Q_1$ for all $l \in [i,j]$. Since this is true for all consecutive latitudes, we see that $Q_1$ can be written as a disjoint union of at most $k$ latitude spherical rectangles.  An identical argument proves the result for $Q_2$.
\end{proof}

Using the above claim, we bound $\l|\frac{|P \cap Q_i|}{|P|}-\frac{\mu(Q_i)}{\mu(\S^2)}\r|$ for $i=1$. An identical bound follows for $i=2$. 

Let $Q_1 = R_1 \cup R_2 \cdots \cup R_j$, $j\le k$, where $R_i$'s are disjoint latitude spherical rectangles. We have 
\begin{align*}
\left|\frac{|P \cap Q_1|}{|P|}-\frac{\mu(Q_1)}{\mu(\S^2)} \right| &= \left|\sum_{i=1}^{j} \left( \frac{|P \cap R_i|}{|P|} -  \frac{\mu(R_i)}{\mu(\S^2)}\right)\right| \\
          									&\le \sum_{i=1}^{j}  \left| \frac{|P \cap R_i|}{|P|} -  \frac{\mu(R_i)}{\mu(\S^2)}\right| \\
									&\le k\D(P)
\end{align*}

We shall now bound $|\mu(Q_1)-\mu(X)|$. The proof for $Q_2$ is identical. The following fact is well known. 
\begin{fact}[\cite{geombookold}]\label{clm:perimeter}The perimeter of a spherical polygon is at most $2 \pi$.\end{fact} 
 Note that since the diameter (the farthest distance between any two points) of the small latitude spherical rectangles is $O(1/k)$, we have that any point in $X$ is at most $O(1/k)$ far from $Q_1$. Now, $Q_1$ contains the following body $Q$. $Q$ is formed by the intersection of sides resulting from shifting each side of $X$ in a perpendicular way by the smallest distance $\de$ inwards such that $Q \subseteq Q_1$. Thus, we have $\de=O(1/k)$. Clearly, $Q \subseteq Q_1 \subseteq X$. Also,  $|\mu(Q_1)-\mu(X)|\leq |\mu(Q)-\mu(X)|$. Now, by construction, if $l$ is the perimeter of $X$, then $|\mu(Q)-\mu(X)|=O(l/k)=O(1/k)$ using Fact \ref{clm:perimeter}. The lemma now follows by minimizing the two errors by choosing $k = \frac{1}{\sqrt{\D(P)}}$. 
\end{proof}
Our construction in fact easily generalizes to yield low discrepancy point samplings in $\S^2_{[\theta_1,\theta_2],[\phi_1,\phi_2]}$ against all convex sets, where (using spherical co-ordinates) $$\S^2_{[\theta_1,\theta_2],[\phi_1,\phi_2]} = \{ (1,\theta,\phi) \in \S^2: \theta \in [\theta_1,\theta_2], \phi \in [\phi_1,\phi_2] \} $$
\begin{cor}[bounded range]\label{cor:s2} For all constants $\theta_1,\theta_2,\phi_1,\phi_2$ such that $0 \le \theta_1 < \theta_2 \le \pi$, $0 \le \phi_1 < \phi_2 \le 2\pi$, there exists an efficient deterministic sampling of $N$ points in $\S^2_{[\theta_1,\theta_2],[\phi_1,\phi_2]}$ with discrepancy $\O\l(\sqrt{\log N/N}\r)$ with respect to all convex sets.\\
\end{cor} 
\begin{proof}[Proof sketch] It can be shown that the  pre-image of  $\S^2_{[\theta_1,\theta_2],[\phi_1,\phi_2]}$ under the  Lambert equal-area projection map $\phi$ is an axis-aligned rectangle in $[0,1] \times [0,1]$. The result now follows by using an appropriate scaling of the point collection from Theorem $\ref{thm:mains2gen}$.
\end{proof}
\subsection{$\SO(3):$ The group of rotations in $3$ dimensions}
We now come to the second contribution of this work. We construct a low discrepancy point set with respect to a fairly general class.

Using the collection of low discrepancy points samplings developed in the previous subsection, we now construct a set of low discrepancy point samplings in $\SO(3)$. We shall first need some preliminaries.

\paragraph{Fiber bundles and local Cartesian product.}
We shall first introduce the notion of a local Cartesian product and a fiber bundle. We closely follow \cite{hopflec} while stating the definitions. A map $\pi: E \rightarrow B$ is a fiber bundle with fiber $F$ if, $\forall b \in B$, there is an open set $U \ni b$ such that $\pi^{-1}(U)\stackrel{\sim}{=}F \otimes U$ (Thus, $E$ is a local Cartesian product). In fact, every \emph{fiber} $\pi^{-1}(b)$ is homeomorphic to $F$. Here $B$ is called the base space. Equivalently, we say $E$ is the total space of a fiber bundle $\pi:E \rightarrow B$ with fiber $F$, or $E=F \tilde{\otimes} B$. That is, $E=\cup_{x \in B} F_x$ where each $F_x \stackrel{\sim}{=}F$. Any $F' \subseteq F$ and $B' \subseteq B$ would inherit the local Cartesian product in a natural way; that is, $F' \tilde{\otimes} B'= \cup_{x \in B'} F'_x$ where each $F_x \supseteq F'_x \stackrel{\sim}{=}F' \subseteq F$.

\paragraph{Example 1:} Let $E=B \otimes F$ (that is, a Cartesian product) and $\pi$ be the projection on the first component. Clearly, $E$ is a fiber bundle over $B$. This is called a trivial bundle.

\paragraph{Example 2:} The Mobius strip $E$ is the total space of a fiber bundle $\pi:E \rightarrow \S^1$ with fiber $[0,1]$. This is a nontrivial fiber bundle. 
Thus, $E=\cup_{x \in \S^1} F_x$ where each $F_x \stackrel{\sim}{=}[0,1]$ but the $F_x$ have different angles which is why there is a global twist. If they were all parallel, it would be a Cartesian product and we would have a hollow cylinder.

\paragraph{Example 3:} The Hopf fibration for $\SO(3)$ (described in detail later) is typically one of the oldest nontrivial examples of a fibration.

\paragraph{Discrepancy for local Cartesian products.}
We are now ready to state our generalization of discrepancy for Cartesian product spaces \cite{Chazelle04thediscrepancy} to local Cartesian product spaces that have a separable volume element. That is, $d\mu(E)=cd\mu(F)d\mu(B)$ whenever $E=F \tilde{\otimes} B$. Note that we can allow any distortion $c$. Also, we assume that each fiber has the same measure to avoid pathological cases since we can have $[a,b] \stackrel{\sim}{=} \mathbb{R}$ but they have different measures.

We first state the result on Cartesian products (for comparison) which is clearly a special case of local Cartesian product. 
\begin{thm}[\cite{Chazelle04thediscrepancy}]\label{thm:cart}
Let $E=F \otimes B$. Let $\X \subseteq 2^F$ and $\Y \subseteq 2^B$ be any collection of sets and define $\C = \cup \{F' \otimes B':F' \in \X, B' \in \Y  \}$. Let $Q, R$ be a point collection such that $\D(Q,\X) \leq \eps_1$ and $\D(R,\Y) \leq \eps_2$. Then \[\D(Q \otimes R, \C) \leq \eps_1+\eps_2\]
\end{thm}
We now state the main theorem of this section. We believe this might be of independent interest.
\begin{thm}\label{thm:localcart}
Let $E=F \tilde{\otimes} B$ be a local Cartesian product with a separable volume element and fibers of same measure. Let $\X \subseteq 2^F$ and $\Y \subseteq 2^B$ be any collection of sets and define $\C = \cup \{F' \tilde{\otimes} B':F' \in \X, B' \in \Y  \}$. Let $Q, R$ be a point collection such that $\D(Q,\X) \leq \eps_1$ and $\D(R,\Y) \leq \eps_2$. Then \[\D(Q \tilde{\otimes} R, \C) \leq \eps_1+\eps_2\]
\end{thm}
\begin{proof}Let $X = F' \stackrel{\sim}{\otimes} B'$ be any set in $\C$. From the definition of local Cartesian products we have,  $X =  \cup_{x \in B'} F'_x$ where each $F_x \supseteq F'_x \stackrel{\sim}{=}F' \subseteq F$. Similarly we have $Q \tilde{\otimes} R$.

Fix an element $x \in B'$. Consider the subsets $F'_x$ and $Q_x$ of the fiber $F_x$. Note that $\mu(F'_x)=\mu(F')$ for all $x$.
Thus we have
 \begin{eqnarray*}
 \D(Q_x,F'_x) &=&  \l|\frac{|Q_x \cap F'_x|}{|Q_x|}-\frac{\mu(F'_x)}{\mu(F_x)}\r| \\  
 			&	&\l|\frac{|Q \cap F'|}{|Q|}-\frac{\mu(F')}{\mu(F)}\r| \\  
				&&\le \epsilon_1
\end{eqnarray*}
Using the above estimate, we now bound the discrepancy of the point collection $P=Q \tilde{\otimes} R$.
\begin{eqnarray*}
\D(P,X) &= &\l|\frac{|P \cap X|}{|P|}-\frac{\mu(X)}{\mu(E)}\r| \\  
               &=&\l| \frac{1}{|R|}\l(\sum_{x \in  R \cap B'} \frac{|Q_x \cap F'_x|}{|Q_x|}\r)-\frac{\mu(X)}{\mu(E)}\r| \\  
               & =&\l|\frac{1}{|R|}\l(\sum_{x \in  R \cap B'} \frac{|Q_x \cap F'_x|}{|Q_x|} - \frac{\mu(F'_x)}{\mu(F_x)} \r)+  \l(\frac{|R \cap B'|\mu(F')}{|R| \cdot \mu(F)} -  \frac{\mu(X)}{\mu(E)}\r)\r| \\   
                &  \leq &\frac{|R \cap B'|}{|R|} \eps_1 +  \l|\frac{|R \cap B'|\mu(F')}{|R| \cdot \mu(F)} -  \frac{\mu(X)}{\mu(E)}\r| \\  
                & =&\frac{|R \cap B'|}{|R|} \eps_1 +  \l(\frac{\mu(F')}{\mu(F)}\r)\l|\frac{|R \cap B'|}{|R|} -  \frac{\mu(B')}{\mu(B)}\r| \\  
                & =&\frac{|R \cap B'|}{|R|} \eps_1 +  \l(\frac{\mu(F')}{\mu(F)}\r)\eps_2 \\  
               & \le & \eps_1 +\eps_2
\end{eqnarray*}
\end{proof}

\paragraph{Discrepancy in $\SO(3)$ and the Hopf fibration.}
In topology, the \emph{Hopf fibration} (or Hopf coordinates), introduced by Heinz Hopf in 1931 \cite{hopf}, is used to describe $\S^3$ in terms of a local Cartesian product of $\S^1$ and $\S^2$ which are much easier to visualize. More precisely, we have $\SO(3) \stackrel{\sim}{=} \S^1 \tilde{\otimes} \S^2$, that is $\SO(3)$ is the total space of a fiber bundle $\pi:\SO(3) \rightarrow \S^2$ with fiber $\S^1$. Think of a point in $\S^2$ as specifying the axis of rotation and correspoding to each point, there is a circle which specifies the rotation around that axis.  

As mentioned in the introduction, the target class for low discrepancy sets we consider is a local Cartesian product of convex sets on $\S^1$ and $\S^2$, which is more general than a strict Cartesian product. We call it the class of local Cartesian convex sets. More precisely we have $\X = \cup \{X_1 \tilde{\otimes} X_2:X_i \ convex,\ i=1,2 \}$.

We restate the first main theorem of this work now.\\ \\
\textbf{Theorem \ref{thm:so3} (Restated). }There exists an efficiently generated collection of points $P$ such that $\D(P,\C)=\O\l(\frac{\log^{2/3} N}{N^{1/3}}\r)$ where $\C$ is the class of local Cartesian convex sets.\\

\begin{proof}
Recall that $\SO(3) = \S^1 \tilde{\otimes} \S^2$. Note that $d\mu(\SO(3))=(1/8)d\mu(\S^1)d\mu(\S^2)$ and each fiber has the same length \cite{Yershova_Jain_Lavalle_Mitchell_2010}. Let $P_1$ be a collection of $N_1$ points arranged on the circle such that any two consecutive points are the same distance apart (Lemma~\ref{lem:s1}). Let $P_2$ be the point collection on $\S^2$ described in the previous section. Then by Theorem \ref{thm:s2}, we have $\L(P_2)=O(\sqrt{\log N_2/N_2}$). Now, our proposed point collection for small discrepancy on $\SO(3)$ is $P_1 \tilde{\otimes} P_2$. 
From Theorem \ref{thm:localcart}, we have
$ \L(P) = O(1/N_1)+O(\log N_2/\sqrt{N_2})$. 
Setting $N_2=(N\log N)^{2/3}$ ($N_1=(N/\log^2 N)^{1/3}$) proves the theorem.
\end{proof}
We can in fact construct low discrepancy point samplings for a more general set.

 Define 
$$\SO(3)_{([\theta_1,\theta_2]),([\theta_1^\prime,\theta_2^\prime],[\phi_1^{\prime},\phi_2^{\prime}])} = \S^1_{[\theta_1,\theta_2]} \tilde{\otimes} \S^2_{[\theta_1^{\prime},\theta_2^{\prime}],[\phi_1,\phi_2]}. $$

\begin{cor}[bounded range]\label{cor:so3}For all constants $\theta_1,\theta_2,\theta_1^{\prime},\theta_2^{\prime},\phi_1,\phi_2$ such that $0 \le \theta_1 < \theta_2 \le 2 \pi$, $0 \le \theta_1 < \theta_2 \le \pi$ and $0 \le \phi_1 < \phi_2 \le 2\pi$, there exists an efficient deterministic sampling of $N$ points in $\SO(3)_{([\theta_1,\theta_2]),([\theta_1^\prime,\theta_2^\prime],[\phi_1^{\prime},\phi_2^{\prime}])}$ with discrepancy $\O\l(\frac{\log^{2/3} N}{N^{1/3}}\r)$ against all local Cartesian convex sets.
\end{cor}

\subsection{$\SE(2):$ The group of translation and rotation in $2$ dimensions}
We consider the equivalent group defined by the cartesian product of $\SO(2)$ and $\T(2)$ where we define the volume element naturally to be the product of the volume elements of the two underlying groups. Then, by the previous results and Theorem~\ref{thm:cart}, we have the following. We point out that this is not new and is merely stated for completeness. 
\begin{thm}\label{thm:se2}There exists an efficient deterministic sampling of $N$ points, $P$ in $\SE(2)$ that satisfies $$\D(P,\C_1)=\tilde{\O}\l(N^{-1/2}\r),$$ $$\D(P,\C_2)=\tilde{\O}\l(N^{-1/3}\r),$$ where $\C_1=\mathcal{I} \otimes \mathcal{R}$ and $\C_2=\mathcal{I} \otimes \L$.\end{thm}
We record an easily generalization of the above theorem. Define
$$\SE(2)_{[\theta_1,\theta_2],[\phi_1,\phi_2]} =  \S^2_{[\theta_1,\theta_2],[\phi_1,\phi_2]} \otimes \T(2). $$

\begin{cor}[bounded range]\label{cor:se2}
For all constants $\theta_1,\theta_2,\phi_1,\phi_2$ such that $0 \le \theta_1 < \theta_2 \le \pi$, $0 \le \phi_1 < \phi_2 \le 2\pi$,  there exists an efficient deterministic sampling of $N$ points, $P$ in $\SE(2)_{[\theta_1^\prime,\theta_2^\prime],[\phi_1^{\prime},\phi_2^{\prime}]}$ that has discrepancy $\tilde{\O}\l(N^{-1/2}\r)$ against  $\mathcal{I} \otimes \mathcal{R}$ and
discrepancy $\tilde{\O}\l(N^{-1/3}\r)$ against $\mathcal{I} \otimes \L$.
\end{cor}
\subsection{$\SE(3):$ The group of translation and rotation in $3$ dimensions}
Finally we consider $\SE(3)$ which captures all the above subgroups. As in the case of $\SE(2)$ we look at the equivalent $\SO(3) \otimes \T(3)$ and use our bounds for the underlying groups and combine them using Theorem~\ref{thm:cart}. We have the following. Let $\mathcal{R}, \L$ be the class of axis aligned cubes and convex sets in $3$ dimensions. Let $\C$ be the class of local Cartesian product of convex sets of $\S^1$ and $\S^2$.
\begin{thm}\label{thm:se3}There exists an efficient deterministic sampling of $N$ points, $P$ in $\SE(3)$ that satisfies $$\D(P,\C_1)=\tilde{\O}\l(N^{-1/4}\r),$$ $$\D(P,\C_2)=\tilde{\O}\l(N^{-1/6}\r),$$ where $\C_1=\C \otimes \mathcal{R}$ and $\C_2=\C \otimes \L$.
\end{thm}
We can in fact efficiently generate low discrepancy point samples in a more general space. 

Define
$$\SE(3)_{([\theta_1,\theta_2]),([\theta_1^\prime,\theta_2^\prime],[\phi_1^{\prime},\phi_2^{\prime}])} =  \SO(3)_{([\theta_1,\theta_2]),([\theta_1^\prime,\theta_2^\prime],[\phi_1^{\prime},\phi_2^{\prime}])}  \otimes \T(3). $$

\begin{cor}[bounded range]\label{cor:se3}
For all constants $\theta_1,\theta_2,\theta_1^{\prime},\theta_2^{\prime},\phi_1,\phi_2$ such that $0 \le \theta_1 < \theta_2 \le 2 \pi$, $0 \le \theta_1 < \theta_2 \le \pi$ and $0 \le \phi_1 < \phi_2 \le 2\pi$,  there exists an efficient deterministic sampling of $N$ points, $P$ in $\SE(3)_{([\theta_1,\theta_2]),([\theta_1^\prime,\theta_2^\prime],[\phi_1^{\prime},\phi_2^{\prime}])}$ that has discrepancy $\tilde{\O}\l(N^{-1/4}\r)$ against $\C \otimes \mathcal{R}$ and
discrepancy $\tilde{\O}\l(N^{-1/6}\r)$ against $\C \otimes \L $.
\end{cor}

\section{Low Discrepancy Point Samplings in product spaces}\label{sec:product}
Let $\U$ be the universe. Let $P \subseteq \U$ be a low discrepancy point collection with respect to a class $\C$. We show how to bootstrap $P$ to get a low discrepancy point samplings in  $\U^{n}$ with respect to $\C^n$. An obvious candidate point collection is the product set $P^n$. It follows from Theorem~\ref{thm:cart} that $\D(P^n,\C^n) \le n \D(P,\C)$. However the size of this point collection is $|P|^{n}$ which increases exponentially in $n$. 
 
We get an almost exponential improvement in the size of the point collection using the following idea. We use low discrepancy sets (of size $m$, say) in the basis groups and take a direct product of the sets to obtain a new set of size $m^n$. Now, we apply another level of discrepancy minimization to further choose a subset of the $m^n$ points to still ensure the right discrepancy. The right tool to make the second step work is a pseudorandom generator for combinatorial rectangles constructed in  \cite{GMRZ_13}. Using the results from the first part of the paper, this leads to a host of nontrivial constructions for low discrepancy sets in product spaces of rigid motion groups.
 
 We begin with some preliminaries.

\subsection{Combinatorial rectangles}

\begin{define}
We say that a function $G:\zo^s \rightarrow \U$ is a \emph{pseudorandom generator} for a class $\cal F$ of functions from $\U$ to $\zo$
with error $\eps$ if,
for every $f \in \cal F$,
\[ |\Pr[f(y) =1] - Pr[f(G(x)) = 1]| \leq \epsilon,\]
where $x$ and $y$ are chosen uniformly at random from $\zo^s$ and $\U$, respectively.
The quantity $s$ is called the \emph{seed length} of~$G$.
\end{define}
Let $[r]$ denote the set $\{1,2,\cdots, r\}$ for any integer $r >0$.
\begin{define}
Let $m,n >0$ be positive integers. A subset of $[m]^n$ of the form $I_1 \otimes I_2 \otimes  \cdots  \otimes I_n$, where $I_j \subseteq [m]$ for all $j \in [n]$ is called a combinatorial rectangle in $[m]^n$.
\end{define}
The following pseudorandom generator for combinatorial rectangles  was constructed in \cite{GMRZ_13}.
\begin{thm}[\cite{GMRZ_13}] There exists an efficient construction of a pseudorandom generator for combinatorial rectangles in $[m]^n$ with seed length $$\O\l((\log \log m)\cdot \log (nm/\eps)\r)+\O\l((\log 1/\eps) (\log \log 1/\eps) (\log \log \log 1/\eps))\r)$$ and error $\epsilon$.
\end{thm}
\begin{cor}\label{cor:crprg} There exists an efficiently constructible collection of points in $[m]^n$ of size $$\l(mn/\eps\r)^{\O\l(\log \log m+(\log \log 1/\eps)(\log \log \log 1/\eps)\r)}$$ which has discrepancy $\epsilon$ against the class of combinatorial rectangles in $[m]^n$.
\end{cor}

\subsection{Derandomization using Combinatorial rectangles}

In this section we prove a general de-randomization result to construct low discrepancy point collections in product sets which beats the trivial tensor product construction described above.

 Let  $\U$ be an universe.  For $i \in [n]$,  let  $Q_i$ be a collection of points of size $m$ with discrepancy $\epsilon_i$ against a collection of sets $\X_i$ in $2^{\U}$. Define the class: $$\X_{[n]}=\{ (X_1,X_2,\cdots,X_n): X_i \in \X_i \} $$ We construct low discrepancy points for this class.
 
\begin{thm}\label{thm:genproduct}   For $i \in [n]$, let $\sigma_i : [m] \rightarrow Q_i$ be arbitrary efficiently computable injective maps. Further let $P_{R}$ be the point collection with discrepancy $\epsilon_{R}$ against combinatorial rectangles in $[m]^n$ that we get from Corollary \ref{cor:crprg}. For any point $p_{R} = (a_1, \ldots,a_n) \in P_{R} \subset [m]^n$, define \linebreak $\sigma(p_{R})=(\sigma_1(a_1),\ldots,\sigma_i(a_i),\ldots,\sigma_n(a_n)) \in Q_1 \otimes Q_2 \ldots \otimes Q_n$. Also define the set  $Q_{[n]}= (Q_1 \otimes \ldots \otimes Q_n) \circ P_{R} = \{ \sigma(p_{R}) : p_{R} \in P_{R}\}$.

Then,
$$\Di(Q_{[n]},\X^n) \le \epsilon_{R} + \sum_{i=1}^n \epsilon_i$$

\end{thm}

\begin{proof}
Fix any $X_{[n]} =  X_1 \otimes  \ldots \otimes X_n $, where $X_i \in \X_i$.
\begin{align*}
\Di(Q_{[n]},X_{[n]}) &=   \left|\frac{| Q_{[n]} \cap X_{[n]} |}{|Q_{[n]}|} - \frac{\mu(X_{[n]})}{\mu(\U^n)}\right| \\
                      &=  \left|\frac{| Q_{[n]} \cap ((Q_1 \cap X_1) \otimes 
                      			 \ldots (Q_{n} \cap X_{n})  |}{|Q_{[n]}|} - \frac{\mu(X_{[n]})}{\mu(\U^n)}\right| \ \ \text{(since $ Q_{[n]} \subseteq Q_1 \otimes \ldots \otimes Q_n$)} \\
			&=  \left|\frac{| Q_n \cap  \prod_{i=1}^{n}(Q_i \cap X_i)|}{|Q_{[n]}|} -   \frac{\prod_{i=1}^{n}|Q_i \cap X_i|}{\prod_{i=1}^n|Q_i|} +    \frac{\prod_{i=1}^{n}|Q_i \cap X_i|}{\prod_{i=1}^n|Q_i|} - \frac{\mu(X)}{\mu(\U^n)}\right| \\
			&\le \underbrace{\left|\frac{| Q_n \cap  \prod_{i=1}^{n}(Q_i \cap X_i)|}{|Q_{[n]}|} -   \frac{\prod_{i=1}^{n}|Q_i \cap X_i|}{\prod_{i=1}^n|Q_i|}\right|}_\textrm{A} +
			\underbrace{\left|\frac{\prod_{i=1}^{n}|Q_i \cap X_i|}{\prod_{i=1}^n|Q_i|} - \frac{\mu(X)}{\mu(\U^n)}\right|}_\textrm{B} \\
\end{align*}
We now bound the terms $A$ and $B$. \newline
 Let $X_i^{\prime} = Q_i \cap X_i$ and $R_i = \sigma(X_i^{\prime}) \subseteq [m]$ for each $i \in [n]$.  \begin{claim} $A \le \epsilon_{R}$.
 \end{claim}
 \begin{proof}
 Recalling that each $\sigma_i$ is an injective map, we have
 \begin{align*}
 A &= \left|\frac{| Q_{[n]} \cap  \prod_{i=1}^{n}(Q_i \cap X_i)|}{|Q_{[n]}|} -   \frac{\prod_{i=1}^{n}|X_i^{\prime}|}{\prod_{i=1}^n|Q_i|}\right|  \\
    &=  \left|\frac{| \sigma(Q_{[n]}) \cap  \prod_{i=1}^{n}\sigma(X_i^{\prime})|}{|\sigma(Q_n)|} -   \frac{\prod_{i=1}^{n}|\sigma(X_i^{\prime})|}{\prod_{i=1}^n|Q_i|}\right|  \\
    &= \left| \frac{|P_{R} \cap \prod_{i=1}^{n}R_i| }{|P_{R}|} - \frac{ \prod_{i=1}^{n}|R_i| }{ m^n} \right| \\
    &\le \epsilon_{R}
 \end{align*}
 where the last inequality follows from the fact that $P_R$ has discrepancy $\epsilon_R$ against combinatorial rectangles.
 \end{proof}
  \begin{claim}\label{claim:cartprod} $B    \le \sum_{i=1}^{n} \epsilon_i$
 \end{claim}
Claim $\ref{claim:cartprod}$ follows directly by observing that $B= \D(Q_{[n]},X_{[n]})$ and using Theorem $\ref{thm:localcart}$ (or the result in \cite{chazelle_discrep_book}). 
Thus $\D(Q_{[n]},\X_{[n]}) \le \epsilon_{R} + \sum_{i=1}^n \epsilon_i$.
\end{proof}
For the sake of completeness we include a direct proof of Claim $\ref{claim:cartprod}$ in the Appendix.
\subsection{Instantiations of the main theorem}
We now state as corollaries, the efficient deterministic construction of various low discrepancy point collections in product spaces. It is easy to see that we can obtain low discrepancy samplings for product spaces the hypertorus $\l(\S^1\r)^n$, $\SO(3)^n$, $\SE(3)^n$, $\SE(2)^n$ and others. In general, we can construct efficiently construct low discrepancy point collections against any mixed products of the base classes.  We instantiate our technique for classes: $\SO(3)^n$, $\SE(3)^n$, the mixed product space $\SO(3) \otimes \SE(2) \otimes  \ldots \otimes \SE(2) \otimes \SO(3)$ and a product space with constraints on allowed rotations.

\paragraph{Discrepancy in $\SO(3)^n$.}
\begin{cor} There is an efficient deterministic sampling of $N$ points $P$ in $\SO(3)^n$ satisfying $\D(P,\C^n) \leq \eps$ where $\C$ is the class of local Cartesian convex sets in $\SO(3)$ and $$N= \l(\frac{n}{\epsilon}\r)^{\O\l(\log \log n+\log \log (1/\eps)\log \log \log (1/\eps)\r)}.$$
\end{cor}
\begin{proof}
The proof follows from Theorem \ref{thm:genproduct} and Theorem \ref{thm:so3}.


\end{proof}

\paragraph{Discrepancy in $\SE(3)^n$.}
\begin{cor} There is an efficient deterministic sampling of $N$ points $P$ in $\SE(3)^n$ satisfying $\D(P,\C^n) \leq \eps$ where $\C$ is the class of product of local Cartesian convex sets in $\SO(3)$ and convex sets in $\R^3$ and $$N= \l(\frac{n}{\epsilon}\r)^{\O\l(\log \log n+\log \log (1/\eps)\log \log \log (1/\eps)\r)}.$$
\end{cor}
\begin{proof}
The proof follows from Theorem \ref{thm:genproduct} and Theorem \ref{thm:se3}.
\end{proof}

\paragraph{Discrepancy in $\SO(3) \otimes \SE(2)\otimes \ldots \otimes \SE(2) \otimes \SO(3)$.}
\begin{cor}There is an efficient deterministic sampling of $N$ points $P$ in \linebreak $\underbrace{\SE(2) \otimes \SO(3) \ldots \otimes \SE(2) \otimes \SO(3)}_{n}$  satisfying $\D(P,\C_1 \otimes \C_2\otimes \ldots \otimes \C_1 \otimes \C_2) \leq \eps$ where $\C_1$ is the class of convex sets in $\SO(3)$ and $C_2$ is the class of product of contiguous intervals and convex sets in $\R_2$ and $$N= \l(\frac{n}{\epsilon}\r)^{\O\l(\log \log n+\log \log (1/\eps)\log \log \log (1/\eps)\r)}.$$
\end{cor}
\begin{proof}
The proof follows from Theorem \ref{thm:genproduct} and Theorem \ref{thm:se2}.
\end{proof}
\paragraph{Discrepancy in $\SO(3)^{([\alpha^{(1)}_1,\alpha^{(1)}_2]),([\theta_1^{(1)},\theta_2^{(1)}],[\phi_1^{(1)},\phi_2^{(1)}])} \otimes  \ldots \otimes \SO(3)^{([\alpha^{(n)}_1,\alpha^{(n)}_2]),([\theta_1^{(n)},\theta_2^{(n)}],[\phi_1^{(n)},\phi_2^{(n)}])}$.}
\begin{cor} For all constants $\alpha^{(i)}_1,\alpha^{(i)}_2,\theta^{(i)}_1,\theta^{(i)}_2,\phi^{(i)}_1,\phi^{(i)}_2$, $i \in \{ 1,\ldots,n\}$, such that $0 \le \alpha^{(i)}_1 < \alpha^{(i)}_2 \le 2\pi$, $0 \le \theta^{(i)}_1 < \theta^{(i)}_2 \le \pi$, $0 \le \phi^{(i)}_1 < \phi^{(i)}_2 \le 2\pi$, there is an efficient deterministic sampling of $N$ points $P$ in $\SO(3)^{([\alpha^{(1)}_1,\alpha^{(1)}_2]),([\theta_1^{(1)},\theta_2^{(1)}],[\phi_1^{(1)},\phi_2^{(1)}])} \otimes  \ldots \otimes \SO(3)^{([\alpha^{(n)}_1,\alpha^{(n)}_2]),([\theta_1^{(n)},\theta_2^{(n)}],[\phi_1^{(n)},\phi_2^{(n)}])}$ with discrepancy $\epsilon$ against $\C^1 \otimes \ldots \otimes \C^n$, $\C^{i}$ is the class of local Cartesian convex sets in \linebreak $\SO(3)^{([\alpha^{(i)}_1,\alpha^{(i)}_2]),([\theta_1^{(i)},\theta_2^{(i)}],[\phi_1^{(i)},\phi_2^{(i)}])}$ and $$N= \l(\frac{n}{\epsilon}\r)^{\O\l(\log \log n+\log \log (1/\eps)\log \log \log (1/\eps)\r)}.$$
\end{cor}
\begin{proof}
The proof follows from Theorem \ref{thm:genproduct} and Theorem \ref{cor:so3}.
\end{proof}


\bibliographystyle{alpha}
\bibliography{sampling}

\begin{thebibliography}{GMRZ13}

\bibitem[ABD12]{ABD_12}
C.~Aistleitner, J.~Brauchart, and J.~Dick.
\newblock Point sets on the sphere $\mathbb{S}^2$ with small spherical cap
  discrepancy.
\newblock {\em Discrete {\&} Computational Geometry}, 48(4):990--1024, 2012.

\bibitem[BBBV13]{baj3}
C.~Bajaj, B.~Bauer, R.~Bettadapura, and A.~Vollrath.
\newblock {Nonuniform Fourier Transforms for Rigid-Body and Multi-Dimensional
  Rotational Correlations}.
\newblock {\em SIAM Journal of Scientific Computing, Accepted for publication},
  35(4):{B821--B845}, 2013.

\bibitem[BC05]{geombookold}
W.~Bush and J.~Clarke.
\newblock {\em The elements of geometry}.
\newblock { New York, Boston [etc.] Silver, Burdett and company},
  url={https://archive.org/details/elementsgeometr01clargoog}, 1905.

\bibitem[BCS11]{baj1}
C.~Bajaj, R.~Chowdhury, and V.~Siddahanavalli.
\newblock {{F2Dock}: Fast Fourier Protein-protein Docking}.
\newblock {\em IEEE/ACM Trans. Comput. Biol. Bioinformatics}, 8(1):45--58,
  2011.

\bibitem[BD11]{bd12}
J.~S. {Brauchart} and J.~{Dick}.
\newblock {Quasi-Monte Carlo rules for numerical integration over the unit
  sphere \$$\backslash$mathbb$\{$S$\}$\^{}2\$}.
\newblock {\em ArXiv e-prints}, January 2011.

\bibitem[CC13]{Chetelat_Chirikjian}
O.~Chetelat and G.S. Chirikjian.
\newblock Sampling and convolution on motion groups using generalized gaussian
  functions.
\newblock {\em Electronic Journal of Computational Kinematics}, accepted, 2013.

\bibitem[Cha00]{chazelle_discrep_book}
B.~Chazelle.
\newblock {\em The Discrepancy Method: Randomness and Complexity}.
\newblock Cambridge University Press, 2000.

\bibitem[Cha04]{Chazelle04thediscrepancy}
B.~Chazelle.
\newblock The discrepancy method in computational geometry.
\newblock In {\em In Handbook of Discrete and Computational Geometry}, pages
  983--996. CRC Press, 2004.

\bibitem[GMRZ13]{GMRZ_13}
P.~Gopalan, R.~Meka, O.~Reingold, and D.~Zuckerman.
\newblock Pseudorandom generators for combinatorial shapes.
\newblock {\em SIAM J. Comput.}, 42(3):1051--1076, 2013.

\bibitem[Hop31]{hopf}
H.~Hopf.
\newblock {Über die Abbildungen der dreidimensionalen Sphäre auf die
  Kugelfläche}.
\newblock {\em Mathematische Annalen}, 104(1):637--665, 1931.

\bibitem[KN74]{KuipNeid}
L.~Kuipers and H.~Niederreiter.
\newblock {\em {Uniform distribution of sequences}}.
\newblock Wiley, New York, 1974.

\bibitem[LL03]{Lind-Valle2003}
S~Lindemann and S.~LaValle.
\newblock Incremental low-discrepancy lattice methods for motion planning.
\newblock In {\em Proceedings IEEE International Conference on Robotics and
  Automation}, pages 2920--2927, 2003.

\bibitem[Mat99]{geombook}
J.~Matousek.
\newblock {\em Geometric Discrepancy}.
\newblock {Algorithms and Combinatorics 18, Springer, Berlin,}, 1999.

\bibitem[Mit08]{Mitchell_2008}
J.~Mitchell.
\newblock Sampling rotation groups by successive orthogonal images.
\newblock {\em SIAM Journal on Scientific Computing}, 30(1):525--547, 2008.

\bibitem[Nie90]{Niederreiter_1990}
H.~Niederreiter.
\newblock Quasi-monte carlo methods.
\newblock {\em Encyclopedia of Quantitative Finance}, 24(1):55--61, 1990.

\bibitem[Nie92]{book:neid}
H.~Niederreiter.
\newblock {\em Random Number Generation and quasi-Monte Carlo Methods}.
\newblock Society for Industrial and Applied Mathematics, Philadelphia, PA,
  USA, 1992.

\bibitem[RST96]{Radovic_Sobol_Tichy_1996}
I.~Radovic, I.~Sobol, and R.~Tichy.
\newblock Quasi-monte carlo methods for numerical integration: Comparison of
  different low discrepancy sequences.
\newblock {\em Monte Carlo Methods and Appl}, 2(1):1--14, 1996.

\bibitem[Sob67]{Sobol67}
I.M Sobol'.
\newblock On the distribution of points in a cube and the approximate
  evaluation of integrals.
\newblock {\em \{USSR\} Computational Mathematics and Mathematical Physics},
  7(4):86 -- 112, 1967.

\bibitem[Tot]{hopflec}
B.J. Totaro.
\newblock Fiber bundles: https://www.dpmms.cam.ac.uk/~bt219/fiber.pdf.

\bibitem[WS08]{Wang_Sloan_2008}
X.~Wang and I.~Sloan.
\newblock Low discrepancy sequences in high dimensions: How well are their
  projections distributed?
\newblock {\em Journal of Computational and Applied Mathematics},
  213(2):366--386, 2008.

\bibitem[YJLM10]{Yershova_Jain_Lavalle_Mitchell_2010}
A.~Yershova, S.~Jain, S.~Lavalle, and J.~Mitchell.
\newblock Generating uniform incremental grids on so(3) using the hopf
  fibration.
\newblock {\em The International Journal of Robotics Research}, 29(7):801--812,
  2010.

\bibitem[YL04]{Yersh-Valle2004}
A.~Yershova and S.~M LaValle.
\newblock Deterministic sampling methods for spheres and so(3).
\newblock In {\em Proceedings IEEE International Conference on Robotics and
  Automation}, 2004.

\end{thebibliography}
\section*{Appendix}

\begin{proof}[Proof of Claim $\ref{claim:cartprod}$]
 Let $| X_i^{\prime}| = a_i$ and $\mu(X_i) = b_i$ for each $i \in [n]$ and $\mu(\U) = k $. Thus we have  $$B  = \left|  \frac{\prod_{i=1}^{n} a_i}{m^n} - \frac{\prod_{i=1}^{n}b_i}{k^n}\right| $$
 
 We use induction on $n$.
 For $n=1$,  we have by construction $$|\frac{a_1}{m} - \frac{b_1}{k}| \le \epsilon_1$$
 Let $n>1$. We recall that $|\frac{a_i}{m} - \frac{b_i}{k}| \le \epsilon_i$ for all $i \in [n]$. 
 
 We have
 \begin{align*}
 B &= \left|  \frac{\prod_{i=1}^{n} a_i}{m^n} - \frac{\prod_{i=1}^{n}b_i}{k^n} \right| \\
     &= \left|  \frac{\prod_{i=1}^{n} a_i}{m^n} - \frac{b_1}{k} \cdot \frac{\prod_{i=2}^{n} a_i}{m^{n-1}} + \frac{b_1}{k}\cdot \frac{\prod_{i=2}^{n} a_i}{m^{n-1}} - \frac{\prod_{i=1}^{n}b_i}{k^n} \right| \\
     &\le \frac{\prod_{i=2}^{n} a_i}{m^{n-1}} \left| \frac{a_1}{m} -\frac{b_1}{k} \right| + \frac{b_1}{k} \left|\frac{\prod_{i=2}^{n} a_i}{m^{n-1}} - \frac{\prod_{i=2}^{n}b_i}{k^{n-1}} \right| \\
     &\le \epsilon_1 + \left|\frac{\prod_{i=2}^{n} a_i}{m^{n-1}} - \frac{\prod_{i=2}^{n}b_i}{k^{n-1}} \right| \\
     &\le  \epsilon_1 + \sum_{i=2}^n \epsilon_i&& \text{(using induction hypothesis) } \\
     &= \sum_{i=1}^n \epsilon_i
 \end{align*}
\end{proof}

\end{document}